\documentclass[11pt]{article}

\usepackage[a4paper,margin=1.2in]{geometry}
\usepackage{amsmath,amssymb,amsthm,mathtools}
\usepackage[dvipsnames]{xcolor}
\usepackage[colorlinks=true,citecolor=Blue,urlcolor=Blue,linkcolor=Blue]{hyperref}
\usepackage[margin,draft]{fixme}
\usepackage{array,enumerate}
\usepackage{natbib}
\usepackage{verbatim}

  \usepackage{tikz}
\usetikzlibrary {shapes}
\usetikzlibrary {arrows}
\usetikzlibrary {positioning}

\title{Bayesian Networks of Density Operators}
\author{Steffen Lauritzen\thanks{\texttt{lauritzen@math.ku.dk}}\\University of Copenhagen \and Piotr Zwiernik\thanks{\texttt{piotr.zwiernik@upf.edu}}\\Universitat Pompeu Fabra}
\date{}

\providecommand{\keywords}[1]{{\small\textbf{Keywords:}} #1}

\newcommand{\graph}{\mathcal{G}}
\renewcommand{\dag}{\mathcal{D}}
\newcommand{\hilb}{\mathcal{H}}
\newcommand{\cliques}{\mathcal{C}}
\newcommand{\sep}{\mathcal{S}}
\newcommand{\trace}{\mathrm{Tr}}
\newcommand{\C}{\mathbb{C}}
\newcommand{\cL}{\mathcal{L}}
\newcommand{\ent}{\operatorname{S}}
\renewcommand{\>}{\rangle}

\newcommand{\gse}{\mbox{\,$\perp\!\!\!\perp_\graph$\,}}
\newcommand{\dse}{\mbox{\,$\perp\!\!\!\perp_\dag$\,}}
\newcommand{\ciph}{\mbox{\,$\perp\!\!\!\perp_Q$\,}}
\newcommand{\cd}{\,|\,}
\newcommand{\parents}{\operatorname{pa}}
\newcommand{\desc}{\operatorname{de}}
\newcommand{\nondesc}{\operatorname{nd}}
\newcommand{\pre}{\operatorname{pr}}
\newcommand{\ske}{\operatorname{ske}}

\newtheorem{thm}{Theorem}[section]
\newtheorem{prop}[thm]{Proposition}
\newtheorem{lem}[thm]{Lemma}
\newtheorem{cor}[thm]{Corollary}

\theoremstyle{definition}
\newtheorem{definition}[thm]{Definition}
\newtheorem{ex}[thm]{Example}
\newtheorem{rem}[thm]{Remark}

\begin{document}

\maketitle

\begin{abstract}
We study quantum analogues of Bayesian networks on a directed acyclic graph
(DAG), distinguishing two constructions for positive definite density
operators on finite-dimensional tensor-product Hilbert spaces. The intrinsic
construction starts from a joint state and its conditional-independence
properties. The extrinsic construction assembles a state sequentially from
prescribed local quantum kernels, following an ordering compatible with the
arrows of the DAG. For the intrinsic construction, we prove the equivalence of
the ordered, local, and global directed Markov properties, together with
entropy, recursive-factorization, and logarithmic characterizations. The
extrinsic construction always gives a normalized state and recovers each
kernel as a conditional on all preceding systems. The same kernel, however,
need not be recovered from the marginal on the vertex and its parents; a
three-qubit example exhibits this obstruction. We prove that independence of
the chosen topological ordering is sufficient exactly for transitive DAGs:
every order-invariant kernel family then yields an intrinsically directed
Markov state. Finally, we associate a logarithmic candidate with every
positive definite state and DAG, prove that it is subnormalized, and show that
the trace-one candidate is a directed Markov state. Both the candidate and the
excess global information are invariant under DAG Markov equivalence.
\end{abstract}


\section{Introduction and summary}

Classical Bayesian networks have two equivalent interpretations. On the one
hand, a directed acyclic graph (DAG) gives a local representation of joint
distributions. Each vertex $v$ is assigned a Markov kernel
$p(x_v\cd x_{\parents(v)})$, and these kernels are combined by the product
formula
$$
    p(x_V)
    =
    \prod_{v\in V}p(x_v\cd x_{\parents(v)}).
$$
On the other hand, the same DAG encodes conditional independence restrictions:
the resulting distribution satisfies the ordered, local and global Markov
properties associated with the DAG.

In the classical setting these two interpretations agree.
The product does not depend on the order in which the factors are
multiplied, and the kernels used to construct the joint distribution are
exactly the conditional distributions of the resulting joint law. Thus the same
graph simultaneously describes a representation by local mechanisms and the
conditional independence structure of the joint distribution. This equivalence
is one of the reasons why DAGs are central in statistical modelling, causal
inference and structural equation systems; see, for example,
\citet{pearl:00}.

There is already a substantial literature on quantum analogues of Bayesian
networks and causal graphical models. Early quantum Bayesian network
formalisms were proposed by \citet{tucci_1997}. The conditional-state approach
of \citet{leifer_poulin_2008} is close to the intrinsic viewpoint
taken here: it develops quantum graphical models using independence
properties and proposes quantum analogues
of Bayesian networks, Markov networks, and belief propagation. Related work by
\citet{leifer_spekkens_2013} studies quantum analogues of Bayesian updating
and Bayesian inversion in the same conditional-state language. A different
line of work, closer to operational quantum causal modelling, represents local
mechanisms by channels or process operators.
\citet[Definition~10]{allen_barrett_horsman_lee_spekkens_2017}, in particular,
define a quantum causal model on a DAG by a product of pairwise commuting local
channel operators. Related operational frameworks include
\citet{costa_shrapnel_2016}, \citet{giarmatzi_costa_2018} and
\citet{barrett_lorenz_oreshkov_2021}.

The present paper studies the relationship between two ways of using local
quantum objects in a directed graphical model. The first is intrinsic: one
begins with a joint density operator and derives the corresponding conditional
density operators.
The second is extrinsic: one specifies local positive definite operators and
uses them to construct a joint density operator. These two ways coincide in the
classical case, but not for quantum density operators. Our main results show that the
intrinsic viewpoint retains the classical equivalence of directed Markov
properties, while the extrinsic construction generally depends on an ordering
and coincides with the intrinsic construction only under additional
compatibility conditions.

The extrinsically specified local positive definite operators will be called
quantum kernels. For a vertex $v$ with parents $\parents(v)$, such a kernel is a
positive definite operator $Q_{v|\parents(v)}$ on
$\hilb_v\otimes\hilb_{\parents(v)}$ satisfying
$$
    \trace_v\!\left(Q_{v|\parents(v)}\right)
    =
    I_{\parents(v)}.
$$

We first develop the intrinsic directed quantum Markov structure. The object
is a positive definite density operator $\rho$ on a tensor-product Hilbert
space. The DAG is interpreted through the quantum conditional independences
satisfied by $\rho$. This places the directed
Markov part of the paper in a tradition similar to
\citet{leifer_poulin_2008}. We establish that the ordered, local, and global
directed Markov properties are equivalent for arbitrary finite DAGs.  The
corresponding excess global information is exactly the sum of the ordered
conditional mutual informations. Consequently, each intrinsic
characterization below applies under any one of the three Markov properties.

The implication from the local Markov property {\rm(L)} to intrinsic
factorization goes back to
\citet[Theorem~4.12]{leifer_poulin_2008}. Specializing their explicit
construction to the direct conditional density operator, corresponding to
$n=1$ in their notation, gives a recursive factorization using
$$
    \rho_{v|\parents(v)}
    =
    \rho_{\parents(v)}^{-1/2}
    \rho_{v\cup\parents(v)}
    \rho_{\parents(v)}^{-1/2}.
$$
The converse direction--from this recursive factorization back to directed
Markovness--is not established there. We prove it by showing that every
intermediate state in the recursion is the corresponding marginal of the
final state. Combined with our equivalence of the ordered, local, and global
Markov properties, this shows that each of the three properties is equivalent
to the intrinsic factorization. This provides the intrinsic benchmark for the
extrinsic kernel construction. Thus, at the level of intrinsic Markov
structure, the main equivalences of the classical theory survive. Along a
topological ordering, both the entropy chain rule and the logarithmic
conditionals telescope one vertex at a time, avoiding the recursive
clique--separator arguments needed in the undirected chordal formulation.

We then compare this intrinsic structure to an extrinsic kernel construction.
Given quantum kernels $Q_{v|\parents(v)}$ and a topological ordering
$v_1<\cdots<v_n$, write $V_j=\{v_1,\ldots,v_j\}$. We then construct a
density operator recursively by
$$
    \omega_{V_j}
    =
    Q_{v_j|\parents(v_j)}\star \omega_{V_{j-1}},
    \qquad \text{with }
    M\star N=N^{1/2}MN^{1/2}.
$$
This is the density-operator analogue of the classical recursive
construction by local kernels. It always produces a normalized state, and the
specified kernel is recovered as the conditional density operator of the new
vertex given the full ordered past:
$$
    \omega_{v_j|V_{j-1}}
    =
    Q_{v_j|\parents(v_j)}
    \otimes I_{V_{j-1}\setminus\parents(v_j)}.
$$
In contrast to classical probability, in the quantum case this does not imply that
$$
    \omega_{v_j|\parents(v_j)}
    =
    Q_{v_j|\parents(v_j)}.
$$
Thus the extrinsic kernel construction is specific to the chosen ordering and
need not produce a directed quantum Markov state.

Thus, in the noncommutative setting, the extrinsic construction is indexed by
a DAG together with a chosen topological ordering. The ordering specifies the
procedural order in which the local kernels are combined. Pairwise
commutativity is the special case in which the ordering disappears: the
$\star$-construction becomes an unordered product and, as
Corollary~\ref{cor:commuting-kernels} makes precise, the extrinsic and intrinsic
constructions coincide.

We also study when the extrinsic construction is independent of the chosen
ordering. A
family of kernels is called order-invariant if it gives the same density
operator for every topological ordering of the DAG. Order-invariance allows one
to compare the full-past conditionals obtained from different orderings and
thereby derive quantum conditional independences. We prove that transitivity is
the graph-theoretic condition under which order-invariant extrinsic kernel
constructions automatically have the intrinsic directed Markov interpretation.
For transitive DAGs, order-invariance forces directed quantum Markovness. For
non-transitive DAGs, there exist order-invariant families of positive definite
quantum kernels whose common state is not directed quantum Markov.

A complementary additive viewpoint is provided by logarithmic conditionals
$$
    h_{A|B}(\rho)=\log\rho_{A\cup B}-\log\rho_B.
$$
These are generally not logarithms of the direct conditional density
operators, but they linearize the Markov identities. A positive definite state
is directed quantum Markov with respect to a DAG if and only if
$$
    \log\rho
    =
    \sum_{v\in V}
    \left(
        \log\rho_{v\cup\parents(v)}
        -
        \log\rho_{\parents(v)}
    \right).
$$
Thus directed quantum Markovness is additive in logarithmic coordinates, giving
a noncommutative analogue of the classical identity
$$
    \log p(x_V)=\sum_{v\in V}\log p(x_v\cd x_{\parents(v)}).
$$
Moreover, for any positive definite state $\sigma$, the same formula defines
a logarithmic candidate
$$
    T_\dag(\sigma)
    =
    \exp\left\{
    \sum_{v\in V}
    \left(
        \log\sigma_{v\cup\parents(v)}
        -
        \log\sigma_{\parents(v)}
    \right)
    \right\}.
$$
We prove the two-sided trace bound
$$
    e^{-gI(\dag)_\sigma}
    \leq
    \trace\!\left(T_\dag(\sigma)\right)
    \leq
    1,
$$
where $gI(\dag)_\sigma$ is the excess global information of $\sigma$ relative
to $\dag$, defined in \eqref{eq:excess-global-information}. If
$\trace\!\left(T_\dag(\sigma)\right)=1$, then
$T_\dag(\sigma)$ is itself a directed quantum Markov state. The original state
$\sigma$ is directed quantum Markov if and only if
$T_\dag(\sigma)=\sigma$. This is the directed analogue of the
logarithmic-candidate viewpoint developed for quantum Markov completions in
\citet{lauritzen:zwiernik:26}.

The two state-based quantities introduced here also respect DAG Markov
equivalence: separation-equivalent DAGs have the same excess global
information and the same logarithmic candidate for every positive definite
input state, including non-Markov states.  For perfect orientations of a
chordal graph, this identifies the directed construction with the
clique--separator construction for the corresponding undirected model.

The paper is organized as follows. Section~\ref{sec:preliminaries} introduces
density operators, partial traces, conditional density operators, entropy,
logarithmic conditionals, and quantum conditional independence.
Section~\ref{sec:dag-markov} establishes the equivalence of the ordered, local, and
global directed Markov properties, derives the directed entropy relation, and
proves the intrinsic factorization theorem. Section~\ref{sec:extrinsic-kernels}
develops the extrinsic ordered kernel construction, illustrates the gap between
extrinsic kernels and intrinsic Markov structure, and studies order-invariant
kernel families. Section~\ref{sec:logarithmic-directed} gives the logarithmic
characterization and studies the associated logarithmic candidate.
Section~\ref{sec:markov-equivalence} relates directed and undirected quantum
Markov structures and establishes invariance under DAG Markov equivalence.

\section{Preliminaries}\label{sec:preliminaries}

In this section we describe our notation and collect a few elementary facts that will be
used repeatedly. Since several of our later arguments rely
on concrete matrix manipulations, we keep the presentation explicit.

\subsection{Basic setup}

Let $V$ be a finite set and let $\hilb_v$, $v\in V$, be finite-dimensional
complex Hilbert spaces. For $A\subseteq V$, write
$\hilb_A=\bigotimes_{v\in A}\hilb_v$, with the convention
$\hilb_\varnothing=\C$.
Let $\cL(\hilb_A)$ denote the space
of linear operators on $\hilb_A$. We equip $\cL(\hilb_A)$ with the
Hilbert--Schmidt inner product
$\<M,N\>:=\trace(MN^*)$. Write $\mathbb S(\hilb_A)$ for the real vector space of self-adjoint operators
on $\hilb_A$, $\mathbb S^+(\hilb_A)$ for the cone of positive definite
operators, and
$$\mathbb S_1^+(\hilb_A):=
\{\rho\in\cL(\hilb_A):\rho=\rho^*,\ \rho\succ0,\ \trace(\rho)=1\}$$
for the set of positive definite density operators.  Whenever an operator on
$\hilb_A$ is used on a larger tensor product, the identity on the
complementary subsystem is understood.
We shall use the following operation on the positive definite cone.

\begin{definition}
For $M,N\in\mathbb S^+(\hilb)$, define
$$
    M\star N:=N^{1/2}MN^{1/2}.
$$
\end{definition}
The operation $\star$ is generally neither commutative nor associative. It
will be used throughout for conditional reconstruction. Moreover,
$M\star N=MN$ if and only if $M$ and $N$ commute. As in
\citet{lauritzen:zwiernik:26}, we have for notational convenience reversed the
roles of $M$ and $N$ relative to \citet{leifer_poulin_2008}.

\subsection{Partial trace}

The quantum analogue of marginalization is the partial trace. Thus, whenever
we speak of a marginal of a density operator, we mean the corresponding
reduced density operator obtained by tracing out the complementary subsystem.
We recall the basic facts we need; see, for example,
\citet[Section~2.4.3]{nielsen:chuang:00}.

Let $A,B\subseteq V$ be disjoint finite sets. The \emph{partial trace over $A$} is the
unique linear map
$\trace_A(\,\cdot\,):\cL(\hilb_{A\cup B})\to \cL(\hilb_B)$
such that
$$
\trace\bigl((I_A\otimes M)\rho\bigr)
=
\trace\bigl(M\,\trace_A(\rho)\bigr)
$$
for all operators $\rho$ and $M$ on the indicated spaces. Equivalently, under
the identification
$\cL(\hilb_{A\cup B})\cong \cL(\hilb_A)\otimes \cL(\hilb_B)$,
it is the linear map determined by
$$
\trace_A(X\otimes Y)=\trace(X)\,Y,
\qquad
X\in\cL(\hilb_A),\ Y\in\cL(\hilb_B).
$$

The partial trace is linear, positive, and trace-preserving.  We use the terms
\emph{marginal} and \emph{reduced density operator} interchangeably.
Marginalization is consistent: if $D\subseteq E\subseteq V$, then
$$
    \trace_{E\setminus D}\!\left(\trace_{V\setminus E}(\rho)\right)
    =
    \trace_{V\setminus D}(\rho).
$$

We shall also use the following standard pull-out property; see
\citet[Lemma~2.3]{lauritzen:zwiernik:26}.

\begin{lem}[Pull-out property]\label{lem:pull-out}
Let $A,B,C$ be pairwise disjoint, let $\rho\in\cL(\hilb_{A\cup B\cup C})$, and
let $M\in\cL(\hilb_{B\cup C})$. Then
$$\trace_A\bigl((I_A\otimes M)\rho\bigr)=M\,\trace_A(\rho),\qquad
\trace_A\bigl(\rho(I_A\otimes M)\bigr)=\trace_A(\rho)\,M.$$
\end{lem}

\subsection{Entropy and logarithmic conditionals}

We use the following information-theoretic quantities. The
von Neumann entropy of a density operator $\rho\in \mathbb S_1^+(\hilb)$ is
$$
    \ent(\rho):=-\trace(\rho\log\rho).
$$
For positive definite density operators $\rho$ and $\sigma$ on the same
Hilbert space, their quantum relative entropy is
$$
    D(\rho\|\sigma)
    :=
    \trace\!\left(\rho(\log\rho-\log\sigma)\right).
$$
Klein's inequality states that $D(\rho\|\sigma)\geq0$, with equality if and
only if $\rho=\sigma$.
For a state $\rho$ and $A\subseteq V$, we write
$$
    \ent(A)_\rho:=\ent(\rho_A).
$$
We use the convention $\ent(\varnothing)_\rho=0$.
If $A$ and $B$ are disjoint, the quantum conditional entropy is
$$
    \ent(A\cd B)_\rho
    :=
    \ent(A\cup B)_\rho-\ent(B)_\rho.
$$
The mutual information is
$$
    I(A:B)_\rho
    :=
    \ent(A)_\rho+\ent(B)_\rho-\ent(A\cup B)_\rho.
$$
For the logarithmic constructions below, we use a conditional object adapted
to additive identities. If $\rho_{A\cup B}$ is positive definite, we define the
\emph{logarithmic conditional} of $A$ given $B$ by
\begin{equation}\label{eq:logcond}
    h_{A|B}(\rho)
    :=
    \log\rho_{A\cup B}-\log\rho_B,
\end{equation}
where $\log\rho_B$ is embedded into $\cL(\hilb_{A\cup B})$ by tensoring with
the identity on $A$.

In the classical case, $h_{A|B}=\log p(a,b)-\log p(b)$ is the logarithm of
the usual conditional probability.
The following entropic identity follows from the definition. Since
$\trace_A(\rho_{A\cup B})=\rho_B$, we have
\begin{equation}\label{eq:logcond-expectation}
    \trace\bigl(\rho_{A\cup B}h_{A|B}(\rho)\bigr)
    =
    \trace(\rho_{A\cup B}\log\rho_{A\cup B})
    -
    \trace(\rho_B\log\rho_B)
    =
    -\ent(A\cd B)_\rho.
\end{equation}
Thus the expectation of a logarithmic conditional is the negative conditional
entropy.

\subsection{Quantum conditional independence}
\label{subsec:qci}

Let $A,B,C\subseteq V$ be pairwise disjoint and let
$\rho\in\mathbb S_1^+(\hilb_{A\cup B\cup C})$. The \emph{quantum conditional
mutual information} is
$$
    I(A:B\cd C)_\rho
    :=
    \ent(A\cup C)_\rho
    +
    \ent(B\cup C)_\rho
    -
    \ent(C)_\rho
    -
    \ent(A\cup B\cup C)_\rho.
$$
Equivalently, conditional mutual information is the decrease in conditional
entropy produced by conditioning additionally on $B$:
\begin{equation}\label{eq:cmi-conditional-entropy}
    I(A:B\cd C)_\rho
    =
    \ent(A\cd C)_\rho-\ent(A\cd B\cup C)_\rho.
\end{equation}
This quantity is always
nonnegative by strong subadditivity \citep{lieb_ruskai_1973}.

\begin{definition}
We say that $A$ and $B$ are \emph{quantum conditionally independent given $C$}
with respect to $\rho$ if $I(A:B\cd C)_\rho=0$. In this case we write
$$
    A\ciph B\cd C\,[\rho],
$$
or just $A\ciph B\cd C$ when there is no ambiguity.
\end{definition}

Quantum conditional independence satisfies the semi-graphoid axioms, as shown
in \citet{leifer_poulin_2008}. Thus:
\begin{enumerate}[(Q1)]
    \item $A\ciph B\cd C\implies B\ciph A\cd C$ \emph{(symmetry)};
    \item $A\ciph B\cd C$ and $D\subseteq B$ imply
    $A\ciph D\cd C$ \emph{(decomposition)};
    \item $A \ciph (B \cup D) \cd C\implies
    A \ciph B \cd (C\cup D)$ \emph{(weak union)};
    \item $A\ciph B \cd C$ and $A \ciph D \cd (B\cup C)$ imply
    $A \ciph (B\cup D) \cd C$ \emph{(contraction)}.
\end{enumerate}
In fact, these properties follow directly from strong subadditivity of the
von Neumann entropy \citep[Lemma~5.1]{studeny:05}.

If the relevant density
operators are positive definite, $\ciph$ also satisfies intersection, that is,
it is a graphoid independence model satisfying
$$
\text{\rm{(Q5)}}\quad
A\ciph B \cd (C \cup D)
\text{ and }
A\ciph D\cd (B\cup C)
\implies
A\ciph (B\cup D) \cd C,
$$
see \citet[Proposition~2.5]{lauritzen:zwiernik:26}. The equivalence of the
directed Markov properties and the entropy criterion below use only the
semi-graphoid properties. More generally, the results in the present paper
extend to density operators that are only positive semidefinite by keeping
proper track of their supports. The only changes are the usual support
projections in the conditional-operator identities and the limiting forms of
the entropy inequalities. We omit the details.

\subsection{Conditional density operators and quantum kernels}

For disjoint $A,B\subseteq V$ and
$\rho\in\mathbb S_1^+(\hilb_{A\cup B})$ we define the \emph{direct conditional density
operator}
\begin{equation}\label{eq:direct-conditional}
\rho_{A|B}:=
\rho_B^{-1/2}\rho_{A\cup B}\rho_B^{-1/2}.
\end{equation}
This is one of several possible noncommutative conditional objects; see
\citet{leifer_poulin_2008}. Here it separates a joint state into a conditional
part and a marginal part:
\begin{equation*}
\rho_{A\cup B}=\rho_{A|B}\star\rho_B,
\end{equation*}
where we recall that
$M\star N:=N^{1/2}MN^{1/2}$.
The pull-out property gives $\trace_A(\rho_{A|B})=I_B$, so the conditional
operator has the normalization of a classical Markov kernel.  We call any
positive definite operator $Q_{A|B}$ on $\hilb_{A\cup B}$ satisfying
$\trace_A(Q_{A|B})=I_B$ a \emph{quantum kernel} from $B$ to $A$.

After tensoring with the appropriate identities, such a kernel gives a
state-extension rule through the $\star$-product.  This rule is nonlinear in
the input state and should not be confused with the usual linear notion of a
quantum channel.
\begin{lem}\label{lem:conditional-reconstruction}
Let $A,B,C$ be pairwise disjoint, let $Q_{A|C}$ be a quantum kernel, and let
$\sigma_{B\cup C}$ be a positive definite density operator. Define
\begin{equation}\label{eq:margcondcombine}
\omega:=Q_{A|C}\star\sigma_{B\cup C},
\end{equation}
where both factors are embedded into $\hilb_{A\cup B\cup C}$ in the usual way.
Then
$\omega\in \mathbb S_1^+(\hilb_{A\cup B\cup C})$ is a density operator with
$$
\omega_{B\cup C}=\sigma_{B\cup C}
\qquad\text{and}\qquad
\omega_{A|B\cup C}=Q_{A|C},
$$
where the second identity uses the usual embedding of $Q_{A|C}$ into
$\hilb_{A\cup B\cup C}$.
\end{lem}

\begin{proof}
Since $\sigma_{B\cup C}$ acts only on the subsystem not traced out, the
pull-out property gives
$$
\trace_A(X\star\sigma_{B\cup C})
=
\trace_A(X)\star\sigma_{B\cup C}
$$
for every positive operator $X$ on $\hilb_{A\cup B\cup C}$. Hence
$$
\omega_{B\cup C}
=
\trace_A(Q_{A|C}\star\sigma_{B\cup C})
=
\trace_A(Q_{A|C})\star\sigma_{B\cup C}
=
(I_B\otimes I_C)\star\sigma_{B\cup C}
=
\sigma_{B\cup C}.
$$
Thus $\omega$ has trace one, and it is positive definite by construction.
Finally,
$$
\omega_{A|B\cup C}
=
\sigma_{B\cup C}^{-1/2}\omega\,\sigma_{B\cup C}^{-1/2}
=
Q_{A|C},
$$
again with the usual embedding.
\end{proof}

Conditional independence can be identified through either of the two
conditional objects introduced above; see Theorems~3.3, 3.5, and~3.6 of
\citet{leifer_poulin_2008}.
\begin{prop}\label{prop:one-sided-identities-fixed-state}
Let $\rho\in\mathbb S_1^+(\hilb_{A\cup B\cup C})$. The following statements
are equivalent:
\begin{enumerate}
    \item [(i)] $A\ciph B\cd C$,
    \item [(ii)] $\rho_{A|B\cup C}=\rho_{A|C}$,
    \item [(iii)] $h_{A|B\cup C}(\rho)=h_{A|C}(\rho)$.
\end{enumerate}
\end{prop}

\begin{proof}
The identity $\rho_{A|B\cup C}=\rho_{A|C}$ is equivalent, after multiplying
on the left and right by
$\rho_{B\cup C}^{1/2}$, to
$$
    \rho_{A\cup B\cup C}
    =
    \rho_{B\cup C}^{1/2}
    \rho_C^{-1/2}
    \rho_{A\cup C}
    \rho_C^{-1/2}
    \rho_{B\cup C}^{1/2}.
$$
This is the Petz recovery characterization of equality in monotonicity of
relative entropy under partial trace; see, for example,
\citet[Proposition~A.11]{lauritzen:zwiernik:26}.  This proves
{\rm(i)}$\Leftrightarrow${\rm(ii)}.  By \eqref{eq:logcond}, the equivalence
of {\rm(i)} and {\rm(iii)} is the logarithmic equality condition for strong
subadditivity
\citep[Theorem~1]{ruskai_2002}:
$$
    A\ciph B\cd C
    \quad\Longleftrightarrow\quad
    \log\rho_{A\cup B\cup C}
    -
    \log\rho_{B\cup C}
    =
    \log\rho_{A\cup C}
    -
    \log\rho_C,$$
  completing the proof.
\end{proof}

\begin{rem}
Lemma~\ref{lem:conditional-reconstruction} gives
$\omega_{A|B\cup C}=Q_{A|C}$, but in general
$\omega_{A|C}\neq Q_{A|C}$.  Proposition~\ref{prop:one-sided-identities-fixed-state}
then shows why the extension need not satisfy $A\ciph B\cd C$.  The next
subsection isolates the obstruction.
\end{rem}

\subsection{Criterion for conditional independence}\label{sec:indep-crit}

Let $A,B,C$ be pairwise disjoint, let
$\sigma_{B\cup C}\in\mathbb S_1^+(\hilb_{B\cup C})$, and let $Q_{A|C}$ be
a quantum kernel from $C$ to $A$. Define the state on $A\cup B\cup C$ by
\eqref{eq:margcondcombine}. By
Lemma~\ref{lem:conditional-reconstruction},
$\omega_{B\cup C}=\sigma_{B\cup C}$. The issue is whether the same
conditional object is recovered from the $A\cup C$ marginal.

\begin{prop}[Star pull-out criterion]
\label{prop:star-pull-out-criterion}
With $\omega=Q_{A|C}\star\sigma_{B\cup C}$ as above, the following
statements are equivalent:
\begin{enumerate}[{\rm (i)}]
    \item $\omega_{A|C}=Q_{A|C}$;
    \item
    $
        Q_{A|C}\star \trace_B(\sigma_{B\cup C})
        =
        \trace_B\bigl(Q_{A|C}\star\sigma_{B\cup C}\bigr);
    $
    \item $A\ciph B\cd C$ with respect to $\omega$.
\end{enumerate}
\end{prop}

\begin{proof}
Lemma~\ref{lem:conditional-reconstruction} gives
$\omega_{A|B\cup C}=Q_{A|C}$ and $\omega_C=\sigma_C$. Hence
Proposition~\ref{prop:one-sided-identities-fixed-state} shows that {\rm(iii)}
is equivalent to $\omega_{A|C}=Q_{A|C}$, which is {\rm(i)}. By
\eqref{eq:direct-conditional}, condition {\rm(i)} is in turn equivalent to
$
    \omega_{A\cup C}=Q_{A|C}\star\sigma_C.
$
Since
$\omega_{A\cup C}
=\trace_B\!\left(Q_{A|C}\star\sigma_{B\cup C}\right)$,
we obtain {\rm(ii)}.
\end{proof}

Thus the pull-out identity says that applying the kernel and then removing $B$
gives the same $A\cup C$ marginal as first removing $B$ and then applying the
kernel. Its failure is exactly the local obstruction to quantum conditional
independence.

\begin{cor}
\label{cor:commuting-pull-out}
If $Q_{A|C}$ and $\sigma_{B\cup C}$ commute after embedding into
$A\cup B\cup C$, then the equivalent conditions of
Proposition~\ref{prop:star-pull-out-criterion} hold.
\end{cor}

\begin{proof}
Commutativity gives
$Q_{A|C}\star\sigma_{B\cup C}=Q_{A|C}\sigma_{B\cup C}$.
Taking the partial trace of the commutation relation shows that
$Q_{A|C}$ also commutes with $\sigma_C$. Therefore the pull-out property gives
$$
    \trace_B\!\left(Q_{A|C}\star\sigma_{B\cup C}\right)
    =
    \trace_B\!\left(Q_{A|C}\sigma_{B\cup C}\right)
    =
    Q_{A|C}\sigma_C
    =
    Q_{A|C}\star\sigma_C.
$$
The conclusion follows from
Proposition~\ref{prop:star-pull-out-criterion}.
\end{proof}

\section{Intrinsic directed quantum Markov structure}
\label{sec:dag-markov}
\subsection{Directed acyclic graphs and separation}
A \emph{directed acyclic graph} (DAG) $\dag=(V,E)$ has directed edges and no
directed cycles.  For a vertex $v$, let $\parents(v)$, $\desc(v)$, and
$\nondesc(v)=V\setminus(\{v\}\cup\desc(v))$ denote its parents, descendants,
and non-descendants, respectively.

A non-endpoint vertex $\alpha_i$ on a walk (allowing repeated vertices)
$(a=\alpha_0,\ldots,\alpha_n=b)$ is a \emph{collider} on the walk if
$\alpha_{i-1}\to\alpha_i\leftarrow\alpha_{i+1}$.  The walk is
\emph{active} relative to $S\subseteq V$ if every non-collider on the walk
is outside $S$ and every collider is in $S$.
Otherwise the walk is
\emph{blocked} by $S$. Disjoint subsets $A,B,S\subseteq V$ are
\emph{$d$-separated by $S$} if every walk from $A$ to $B$ is blocked by
$S$; we then write $A\dse B\cd S$. See, for example,
\citet[Chapters~2 and~3]{lauritzen:26} for further details.

For later use, the \emph{skeleton} $\ske(\dag)$ is obtained by replacing every
arrow by an undirected edge and the \emph{moral graph} $\dag^m$ is obtained by
joining every pair of parents of a common child and then taking the skeleton.

\subsection{Equivalent Markov properties and the entropy criterion}
A total ordering $<$ of $V$ is \emph{topological} with respect to a DAG
$\dag$ if every arrow points from lower to higher.  Write
$\pre(v)=\{u\in V:u<v\}$ for the predecessors of $v$.  Whenever the ordering
is written as $v_1<\cdots<v_n$, put
$V_0=\varnothing$ and $V_j=\{v_1,\ldots,v_j\}$; thus
$\pre(v_j)=V_{j-1}$.  Given a DAG $\dag$, a topological ordering $<$, and a
state $\rho$, consider:
\begin{itemize}
\item[(O)] the \emph{ordered Markov property}:
$$
    v\ciph \bigl(\pre(v)\setminus\parents(v)\bigr)\cd\parents(v)
    \qquad (v\in V);
$$
\item[(L)] the \emph{local Markov property}:
$$
    v\ciph
    \bigl(\nondesc(v)\setminus\parents(v)\bigr)\cd\parents(v)
    \qquad (v\in V);
$$
\item[(G)] the \emph{global Markov property}:
$$
    A\dse B\cd S
    \quad\Longrightarrow\quad
    A\ciph B\cd S.
$$
\end{itemize}

These conditions describe Markovness at three different levels: {\rm(O)} uses
a chosen topological ordering, {\rm(L)} uses the parents and non-descendants
of each vertex, and {\rm(G)} uses all $d$-separation statements in the DAG.
Nevertheless, they define the same restriction on quantum density operators:

\begin{thm}[Equivalence of directed Markov properties]
\label{thm:dagmarkeq}
Let $\dag$ be a DAG, let $<$ be a topological ordering, and let
$\rho\in\mathbb S_1^+(\hilb_V)$.  Then {\rm(O)}, {\rm(L)}, and {\rm(G)} are
equivalent.
\end{thm}
\begin{proof}
By Section~\ref{subsec:qci}, quantum conditional independence satisfies the
semi-graphoid axioms.  The standard equivalence theorem for directed Markov
properties in any semi-graphoid independence model therefore applies
\citep{lauritzen:etal:90}.
\end{proof}
Because the global and local properties are independent of the ordering,
Theorem~\ref{thm:dagmarkeq} also shows that the ordered property holds for one
topological ordering if and only if it holds for every topological ordering.
We henceforth call such a state \emph{directed quantum Markov} with respect to
$\dag$.

In analogy with \citet{lauritzen:zwiernik:26}, define the \emph{excess global
information} of $\rho$ relative to $\dag$ by
\begin{equation}\label{eq:excess-global-information}
    gI(\dag)_\rho
    :=
    \sum_{v\in V}\ent(v\cd\parents(v))_\rho-\ent(V)_\rho.
\end{equation}
By \eqref{eq:logcond-expectation}, the same quantity has the logarithmic
representation
\begin{equation}\label{eq:excess-global-information-logarithmic}
    gI(\dag)_\rho
    =
    -\sum_{v\in V}
    \trace\!\left(\rho h_{v|\parents(v)}(\rho)\right)
    +
    \trace\!\left(\rho\log\rho\right).
\end{equation}
The entropy form \eqref{eq:excess-global-information} is convenient for
proving nonnegativity, whereas
\eqref{eq:excess-global-information-logarithmic} is adapted to additive
logarithmic arguments.
Fix a topological ordering $v_1<\cdots<v_n$.  Since
$\ent(v_j\cd V_{j-1})_\rho
=\ent(V_j)_\rho-\ent(V_{j-1})_\rho$, the entropy chain rule telescopes to
\begin{equation}\label{eq:entropy-chain-rule-topological}
    \ent(V)_\rho
    =
    \sum_{j=1}^n\ent(v_j\cd V_{j-1})_\rho.
\end{equation}
Replacing \eqref{eq:entropy-chain-rule-topological} in
\eqref{eq:excess-global-information} and using
\eqref{eq:cmi-conditional-entropy} gives the key decomposition
\begin{equation}\label{eq:ordered-cmi-decomposition}
    gI(\dag)_\rho
    =
    \sum_{j=1}^n
    I\!\left(v_j:V_{j-1}\setminus\parents(v_j)
        \cd\parents(v_j)\right)_\rho.
\end{equation}
Although its summands depend on the chosen topological ordering, their sum
does not.

\begin{thm}[Directed entropy criterion]
\label{thm:entropy_relation}
For every DAG $\dag$ and every
$\rho\in\mathbb S_1^+(\hilb_V)$, $gI(\dag)_\rho\geq0$.
Consequently, $gI(\dag)_\rho=0$ if and only if $\rho$ satisfies any one--and
hence all--of the directed Markov properties {\rm(O)}, {\rm(L)} and {\rm(G)}.
\end{thm}

\begin{proof}
Every term in \eqref{eq:ordered-cmi-decomposition} is nonnegative by strong
subadditivity.  Their sum vanishes exactly when all ordered Markov
restrictions hold, which by Theorem~\ref{thm:dagmarkeq} is equivalent to
directed quantum Markovness.
\end{proof}

\subsection{Intrinsic factorization}

We next relate the Markov properties to an intrinsic factorization in terms
of the state's own conditional density operators.

Fix a topological ordering $<$ of $\dag$ and write it as
$v_1<\cdots <v_n$. For any
$\rho\in\mathbb S_1^+(\hilb_V)$ define, for $j=2,\ldots,n$,
$$
    \rho_{v_j|V_{j-1}}
    :=
    \rho_{V_{j-1}}^{-1/2}\rho_{V_j}\rho_{V_{j-1}}^{-1/2}.
$$
Then $\rho_{V_j}=\rho_{v_j|V_{j-1}}\star\rho_{V_{j-1}}$, and iteration gives
\begin{equation}\label{eq:chain}
        \rho
    =
    \rho_{v_n|V_{n-1}}
    \star
    \bigl(
    \rho_{v_{n-1}|V_{n-2}}
    \star
    \cdots
    \star
    (\rho_{v_2|V_1}\star\rho_{v_1})
    \cdots
    \bigr).
\end{equation}
The parentheses are part of the formula, since $\star$ is not associative in
general.

\begin{definition}[Intrinsic directed factorization]
Let $\dag$ be a DAG, let $<$ be a topological ordering, and let
$\rho\in\mathbb S_1^+(\hilb_V)$. We say that $\rho$
\emph{factorizes intrinsically over $(\dag,<)$} if, writing
$v_1<\cdots<v_n$,
$$
    \rho
    =
    \rho_{v_n|\parents(v_n)}
    \star
    \bigl(
    \rho_{v_{n-1}|\parents(v_{n-1})}
    \star
    \cdots
    \star
    (\rho_{v_2|\parents(v_2)}\star\rho_{v_1})
    \cdots
    \bigr),
$$
where each conditional density operator $\rho_{v|\parents(v)}$ is computed
from the corresponding marginal of $\rho$ and then embedded into the full
ordered past by tensoring with the identity.
\end{definition}

\begin{samepage}
\begin{thm}[Intrinsic factorization theorem]
\label{thm:directed-factorization-equivalence}
Let $\dag$ be a DAG, $<$ a topological ordering, and let
$\rho\in\mathbb S_1^+(\hilb_V)$. Then the following statements are equivalent:
\begin{enumerate}
    \item[(i)] $\rho$ factorizes intrinsically over $(\dag,<)$;
    \item[(ii)] $\rho$ satisfies the ordered Markov property {\rm(O)};
    \item[(iii)] $\rho$ satisfies the local Markov property {\rm(L)};
    \item[(iv)] $\rho$ satisfies the global Markov property {\rm(G)}.
\end{enumerate}
\end{thm}
\end{samepage}

\begin{proof}
By Theorem~\ref{thm:dagmarkeq}, statements {\rm(ii)}--{\rm(iv)} are
equivalent.  If {\rm(ii)} holds, then
Proposition~\ref{prop:one-sided-identities-fixed-state} gives
$\rho_{v_j|V_{j-1}}=\rho_{v_j|\parents(v_j)}$ for every $j$; substitution
in the full-past factorization \eqref{eq:chain} proves {\rm(i)}.

Conversely, assume {\rm(i)}.  Each
$\rho_{v_j|\parents(v_j)}$ is a quantum kernel, and
Lemma~\ref{lem:conditional-reconstruction}, applied successively with
$A=\{v_j\}$, $B=V_{j-1}\setminus\parents(v_j)$, and
$C=\parents(v_j)$, shows that each step preserves the preceding marginal and
recovers the kernel as the conditional of $v_j$ given $V_{j-1}$.  Induction
therefore identifies every intermediate state with the corresponding
marginal of $\rho$ and again gives
$\rho_{v_j|V_{j-1}}=\rho_{v_j|\parents(v_j)}$.  By
Proposition~\ref{prop:one-sided-identities-fixed-state}, these identities are
exactly the ordered Markov restrictions, proving {\rm(ii)}.
\end{proof}

Consequently, intrinsic factorization holds for one topological ordering if
and only if it holds for every topological ordering.

The implication {\rm(iii)}$\Rightarrow${\rm(i)} is the $n=1$ specialization
of \citet[Theorem~4.12]{leifer_poulin_2008}, stated there for the local Markov
property.  Their conditional-density characterizations
\citep[Theorems~3.3, 3.5, and~3.6]{leifer_poulin_2008} underlie
Proposition~\ref{prop:one-sided-identities-fixed-state}.  The converse is the
new part of the theorem: the marginal-preservation argument above shows that
the recursive factors are conditionals of the final state, and hence enforce
all three directed Markov properties.

The next section shows that this order-independence can be lost when kernels
are specified extrinsically.

\section{Extrinsic quantum kernel constructions}
\label{sec:extrinsic-kernels}

\subsection{Ordered kernel construction}

We next pass from a state's own conditional density operators to extrinsically
specified quantum kernels.  Classically, systems of Markov kernels are
equivalent to structural equation models and provide a natural intervention
calculus; see, for example, \citet[Section~4.3.3]{lauritzen:26}.  Every such
classical system obeys the Markov property of its DAG.  The quantum
construction below retains the sequential interpretation but, in general,
does not retain the associated conditional independences.

Fix a topological ordering $v_1<\cdots<v_n$ of $\dag$.
For each root $v$, specify a density operator
$Q_v\in\mathbb S_1^+(\hilb_v)$.  For each non-root $v$, specify a quantum
kernel
$Q_{v|\parents(v)}\in\mathbb S^+(\hilb_{v\cup\parents(v)})$.  When
$v=v_j$, the kernel is embedded into $\hilb_{v_j\cup V_{j-1}}$ by tensoring
with $I_{V_{j-1}\setminus\parents(v_j)}$.

\begin{definition}[Ordered kernel construction]
The ordered kernel construction associated with $(\dag,<)$ maps local
quantum kernels $\mathcal Q=\{Q_{v|\parents(v)}:v\in V\}$ to the state
$\omega=\Phi_{\dag,<}(\mathcal Q)$ defined recursively by
$\omega_{V_1}=Q_{v_1}$ and
$$
    \omega_{V_j}=Q_{v_j|\parents(v_j)}\star\omega_{V_{j-1}},
    \qquad j=2,\ldots,n.
$$
For a root vertex $v$, the notation $Q_{v|\parents(v)}$ means $Q_v$.
\end{definition}

Lemma~\ref{lem:conditional-reconstruction} shows that every intermediate
operator in the recursion is a density operator.

\begin{cor}[Full-past recovery]
\label{cor:ordered-kernel-recovery}
Let $\omega=\Phi_{\dag,<}(\mathcal Q)$ be obtained from the ordered kernel
construction. If $v_j$ is the $j$th vertex in the ordering, then, for
$j=2,\ldots,n$,
$$
    \omega_{v_j|V_{j-1}}
    =
    Q_{v_j|\parents(v_j)},
$$
where the kernel is embedded by tensoring with
$I_{V_{j-1}\setminus\parents(v_j)}$.
\end{cor}

\begin{proof}
At step $j$, apply Lemma~\ref{lem:conditional-reconstruction} with
$A=\{v_j\}$, $B=V_{j-1}\setminus\parents(v_j)$, and
$C=\parents(v_j)$.  It gives
$$
    (\omega_{V_j})_{V_{j-1}}=\omega_{V_{j-1}}
    \qquad\text{and}\qquad
    (\omega_{V_j})_{v_j|V_{j-1}}
    =
    Q_{v_j|\parents(v_j)}.
$$
The first identity shows inductively that $\omega_{V_j}$ is the $V_j$
marginal of the final state.  The second identity then gives the claim.
\end{proof}

The corollary says that the specified kernel is recovered as the conditional
density operator of $v_j$ given the full ordered past. It does not say that the
same kernel is recovered from the marginal on $v_j\cup\parents(v_j)$. In
general we have
$$
    Q_{v_j|\parents(v_j)}\neq \omega_{v_j|\parents(v_j)}.
$$
Thus the state obtained from the ordered kernel construction need not be
directed quantum Markov. The $\star$ pull-out criterion in
Section~\ref{sec:indep-crit} isolates the missing compatibility condition.

\begin{cor}[Commuting kernels]
\label{cor:commuting-kernels}
Let
$\mathcal Q=\{Q_{v|\parents(v)}:v\in V\}$ be a family of quantum kernels. Suppose that their canonical embeddings
into $\cL(\hilb_V)$ commute pairwise. Then, for every topological ordering
$<$ of $\dag$,
$$
    \Phi_{\dag,<}(\mathcal Q)
    =
    \prod_{v\in V}Q_{v|\parents(v)}
    =:\omega
$$
is the same state. Moreover,
$$
    Q_{v|\parents(v)}
    =
    \omega_{v|\parents(v)}
    \qquad (v\in V),
$$
and the extrinsic construction agrees with the intrinsic factorization of
$\omega$. Consequently, $\omega$ satisfies all three directed Markov
properties {\rm(O)}, {\rm(L)} and {\rm(G)}.
\end{cor}

\begin{proof}
Fix a topological ordering $v_1<\cdots<v_n$.  At each step the next kernel
commutes with the current state, so the $\star$-product is the ordinary
product.  This proves both the displayed product formula and its independence
of the ordering.

Now iterate Corollary~\ref{cor:commuting-pull-out}.  At step $j$ it recovers
$Q_{v_j|\parents(v_j)}$ from the parent marginal and gives the corresponding
ordered Markov restriction; later steps preserve that marginal by
Lemma~\ref{lem:conditional-reconstruction}.  Thus the final state has the
ordered Markov property and its extrinsic kernels are precisely its intrinsic
parent conditionals.  The conclusion follows from
Theorems~\ref{thm:dagmarkeq}
and~\ref{thm:directed-factorization-equivalence}.
\end{proof}

\medskip\noindent\emph{Relation to earlier work.}
\citet[Definition~10]{allen_barrett_horsman_lee_spekkens_2017} define a
quantum causal model on a DAG by assigning a local channel operator to each
vertex, requiring these operators to commute pairwise, and taking their
product as the overall operator. For a channel with input $A$ and outputs
$B,C$, their Definition~6 defines quantum conditional independence of the
outputs given the input by the factorization
$\rho_{BC|A}=\rho_{B|A}\rho_{C|A}$. Their Theorem~2 relates this
factorization to compatibility with $A$ as a complete common cause, and their
Theorem~3 shows that it is equivalent to $I(B:C\cd A)=0$ for the channel
operator after normalization to trace one. Their Theorem~4 gives the
corresponding result for several outputs. For general DAGs, they leave open
the problem of finding a property of the overall operator for which
$d$-separation is sound and complete
\citep[Section~VII]{allen_barrett_horsman_lee_spekkens_2017}.

At the algebraic level, Definition~10 is the commuting case of our extrinsic
construction: the local operators have the same partial-trace normalization,
commute pairwise, and define the overall operator by their product. Their
input--output representation is designed to describe interventions, whereas
we regard the product as a joint density operator on the vertex Hilbert
spaces. In the latter setting, Corollary~\ref{cor:commuting-kernels} shows in
addition that the prescribed kernels are recovered as the state's intrinsic
parent conditionals and hence imply the ordered, local and global directed
Markov properties. We do not assert the converse that every intrinsically
Markov state arises from commuting kernels.

\begin{ex}[Classical limit]
Fix a product basis and suppose that every embedded kernel is diagonal in that
basis. Its diagonal entries then have the form
$$
    q_v(x_v\cd x_{\parents(v)}),
    \qquad
    \sum_{x_v}q_v(x_v\cd x_{\parents(v)})=1.
$$
The embedded kernels commute, so Corollary~\ref{cor:commuting-kernels} gives a
diagonal state whose diagonal entries are
$$
    p(x_V)
    =
    \prod_{v\in V}q_v(x_v\cd x_{\parents(v)}).
$$
Thus the commuting construction contains the usual classical Bayesian network
factorization as a special case.
\end{ex}

\subsection{A three-qubit obstruction}
\label{subsec:three-qubit-example}

Corollary~\ref{cor:ordered-kernel-recovery} recovers each kernel after
conditioning on the full ordered past.  Directed Markovness additionally
requires recovery from the parent marginal.  The following example shows
explicitly how tracing out a non-parent predecessor can change that
conditional.

Consider the chain $1\to3\to2$, with qubits at its vertices, and let $X,Z$
denote Pauli matrices
$$X=\begin{pmatrix}0&1\\1&0\end{pmatrix}, \qquad Z=\begin{pmatrix}1&0\\0&-1\end{pmatrix}.$$  For $0<r<1$ and $0<|\theta|<1$, set
$\sigma_{13}=\frac14(I+rX\otimes X)$ and
$Q_{2|3}=\frac12(I+\theta Z\otimes Z)$. Both operators are positive definite
and $\trace_2(Q_{2|3})=I_3$, so
$\omega=(Q_{2|3}\otimes I_1)\star\sigma_{13}$ is a density operator.
Lemma~\ref{lem:conditional-reconstruction} gives the full-past identity
$\omega_{2|1\cup3}=Q_{2|3}\otimes I_1$.

It remains to compute the parent conditional.  The Pauli calculation from
\citet[Lemma~4.2]{lauritzen:zwiernik:26} gives
$\sigma_{13}^{1/2}=\frac12(\alpha I+\beta X\otimes X)$, where
$\alpha^2+\beta^2=1$ and
$\alpha^2-\beta^2=\sqrt{1-r^2}$.  After tracing out system $1$, the mixed
terms vanish because $\trace(X)=0$, while conjugation by $X$ on system $3$
changes the sign of $Z$.  Hence
\begin{equation}\label{eq:three-qubit-parent-conditional}
\begin{aligned}
    \omega_{23}
    &=
    \tfrac12\bigl(\alpha^2Q_{2|3}
        +\beta^2X_3Q_{2|3}X_3\bigr)
    =
    \tfrac14\bigl(I+\theta\sqrt{1-r^2}\,Z\otimes Z\bigr),\\
    \omega_{2|3}
    &=
    2\omega_{23}
    =
    \tfrac12\bigl(I+\theta\sqrt{1-r^2}\,Z\otimes Z\bigr),
\end{aligned}
\end{equation}
where the second line uses $\omega_3=I_3/2$.

Since $\sqrt{1-r^2}<1$, equation
\eqref{eq:three-qubit-parent-conditional} shows that
$\omega_{2|3}\neq Q_{2|3}$ whenever $\theta\neq0$.  Thus the kernel is
recovered from the full past but not from its parent marginal.
\begin{samepage}
The three relevant operators are summarized below.
\begin{center}
\renewcommand{\arraystretch}{1.25}
\begin{tabular}{
    >{\raggedright\arraybackslash}p{0.18\textwidth}
    >{\raggedright\arraybackslash}p{0.35\textwidth}
    >{\raggedright\arraybackslash}p{0.27\textwidth}}
\hline
Quantity & Value & Role \\
\hline
Specified kernel
&
$Q_{2|3}=\frac12(I+\theta Z\otimes Z)$
&
Extrinsic input
\\
Full past
&
$\omega_{2|1\cup3}=Q_{2|3}\otimes I_1$
&
Recovered exactly
\\
Parents only
&
$\omega_{2|3}
=\frac12(I+\theta\sqrt{1-r^2}\,Z\otimes Z)$
&
Reduced correlation
\\
\hline
\end{tabular}
\end{center}
\end{samepage}
By Proposition~\ref{prop:one-sided-identities-fixed-state}, this failure is
equivalent to the failure of $1\ciph2\cd3$.  For instance, taking
$r=\sqrt3/2$ and $\theta=1/2$ reduces the coefficient of
$Z\otimes Z$ from $1/2$ in the prescribed kernel to $1/4$ in the parent
conditional.  This is the chain mechanism used again in
the next subsection and embedded in
Proposition~\ref{prop:nontransitive-obstruction}.

\subsection{Order-invariant kernel constructions}

In a classical Bayesian network, the product
$   \prod_{v\in V} p(x_v\cd x_{\parents(v)})$
does not depend on the order in which the factors are multiplied. Equivalently,
one may construct the joint distribution recursively along any topological
ordering of the DAG and obtain the same result. Thus the topological ordering
is not part of the model.

The situation is different for the ordered quantum kernel construction. The
$\star$-product is noncommutative and nonassociative, and so a choice of
topological ordering is generally part of the construction. In this sense, a
fixed-order construction should be viewed as a quantum analogue of a sequential
procedure, or computational schedule, rather than as a model determined by the
DAG alone.

We therefore ask when this extra ordering information disappears. A family
of quantum kernels
$\mathcal Q=\{Q_{v|\parents(v)}:v\in V\}$ is called order-invariant if
$\Phi_{\dag,<}(\mathcal Q)$ is the same state for every topological ordering
$<$ of $\dag$. In that case we write
$$
    \omega=\Phi_\dag(\mathcal Q)
$$
for the common state.

Order-invariance is the extrinsic-kernel analogue of the order-free classical
Bayesian network factorization. It says that the local
kernels define a state associated with the partial order encoded by the DAG,
rather than with a particular linear extension of that partial order. It also
has consequences for conditional independence because it
allows us to compare the conditionals recovered from different ordered
constructions.

\paragraph{Three motivating graphs}

Before giving the general result, consider the three DAGs in
Figure~\ref{fig:three_graphs}.  They isolate the two ways in which
order-invariance can succeed or fail to enforce the Markov property.
\begin{figure}[htb]
    \centering
    \begin{tikzpicture}
        [node distance = 6mm and 6mm, minimum width = 4mm]
    %% nodes
    \begin{scope}
      \tikzstyle{every node} = [shape = circle,
      font = \scriptsize,
      minimum height = 4mm,
      inner sep = 0pt,
      draw = black,
      fill = white,
      anchor = center,
      text centered]

      \node(a) at (0,0) {$2$};
      \node(b) [right =  of a] {$1$};
      \node(c) [right = of b] {$3$};

      \node(a1) [ right = 10mm of c]{$1$};
      \node (b1) [right = of a1]{$2$};
      \node (c1) [right = of b1]{$3$};
      \node (a2) [right = 10mm of c1]{$1$};
      \node (b2) [right = of a2]{$3$};
      \node (c2)[right = of b2]{$2$};
    \end{scope}

    \begin{scope}[->, > = latex']
      \draw (b) -- (a);
      \draw (b) -- (c);
      \draw (a1) -- (b1);
      \draw (b1)-- (c1);
      \draw (a2) -- (b2);
      \draw (c2) -- (b2);
    \end{scope}

    \begin{scope}
    \tikzstyle{every node} = [node distance = 4mm and 4mm, minimum width = 4mm,
    font = \scriptsize,
      anchor = center,
      text centered]
      \node(afig) [below = 2mm of b]{(a)};
      \node(bfig) [ below = 2mm of b1]{(b)};
      \node (cfig)[below = 2mm of b2]{(c)};
    \end{scope}
\end{tikzpicture}
    \caption{The fork, chain and collider. The vertex labels represent a
    topological ordering in each graph.}
    \label{fig:three_graphs}
\end{figure}
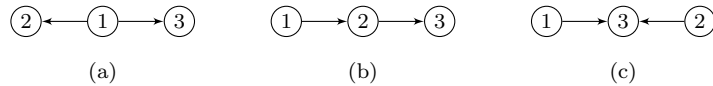

\begin{samepage}
The mechanism is summarized below; the three short arguments that follow
justify each row.
\begin{center}
\renewcommand{\arraystretch}{1.25}
\begin{tabular}{
    >{\raggedright\arraybackslash}p{0.13\textwidth}
    >{\raggedright\arraybackslash}p{0.27\textwidth}
    >{\raggedright\arraybackslash}p{0.38\textwidth}}
\hline
Graph & Available reorderings & Consequence of order-invariance \\
\hline
Fork
&
The two children can be exchanged
&
Their conditionals can be compared, forcing $3\ciph2\cd1$
\\
Chain
&
The ordering is unique
&
Order-invariance is vacuous and need not imply Markovness
\\
Collider
&
The two roots can be exchanged
&
The root factors remain a product, forcing $1\ciph2$
\\
\hline
\end{tabular}
\end{center}
\end{samepage}

\noindent\emph{The fork.}

Consider the fork in Figure~\ref{fig:three_graphs}(a). The two topological
orderings are $(1,2,3)$ and $(1,3,2)$. Suppose that the same state is obtained
from the same kernels in both orders:
$$
    \omega
    =
    Q_{3|1}\star(Q_{2|1}\star Q_1)
    =
    Q_{2|1}\star(Q_{3|1}\star Q_1).
$$
The first ordering gives $\omega_{2|1}=Q_{2|1}$ and
$\omega_{3|12}=Q_{3|1}$. The second ordering gives
$\omega_{3|1}=Q_{3|1}$ and $\omega_{2|13}=Q_{2|1}$. Hence
$\omega_{3|12}=\omega_{3|1}$, and
Proposition~\ref{prop:one-sided-identities-fixed-state} yields
$3\ciph2\cd1$. Thus order-invariance for the fork implies the directed
quantum Markov property.

\medskip\noindent\emph{The chain.}

For the chain in Figure~\ref{fig:three_graphs}(b) there is only one topological
ordering, so every kernel family is trivially order-invariant.  The
construction gives
$$
    \omega=Q_{3|2}\star(Q_{2|1}\star Q_1)
$$
and hence $\omega_{3|12}=Q_{3|2}$. It does not follow that
$Q_{3|2}=\omega_{3|2}$, or equivalently that
$\omega_{3|12}=\omega_{3|2}\otimes I_1$. Thus the kernel construction
need not imply $3\ciph1\cd2$. The explicit Pauli example in
Subsection~\ref{subsec:three-qubit-example} gives such a state. This is the
basic quantum obstruction behind the non-transitive case.

\medskip\noindent\emph{The collider.}

For the collider in Figure~\ref{fig:three_graphs}(c), any topological ordering
places the two roots before the child. A construction has the form
$$
    \omega=Q_{3|12}\star(Q_2\star Q_1)
$$
up to interchanging the two root factors. Since $Q_1$ and $Q_2$ act on
different tensor factors of the Hilbert space, they commute and we have $Q_1\star Q_2=Q_2\star Q_1=Q_1\otimes Q_2$.
Marginalizing over the child gives $\omega_{12}=Q_1\otimes Q_2$, hence
$1\ciph2$. This is exactly the directed Markov property for the collider.

The fork and collider allow incomparable vertices to exchange positions,
whereas the chain has a unique topological ordering. The following result
formalizes this order-comparison mechanism. Let $<$ and $<'$ be two
topological orderings of $\dag$. For a vertex $v$,
write $\pre(v)=\{u:u<v\}$ and $\pre'(v)=\{u:u<'v\}$.

\begin{prop}[Order comparison]\label{prop:order-comparison-qci}
Let $\mathcal Q$ be order-invariant, and let
$\omega=\Phi_\dag(\mathcal Q)$. If, for some vertex $v$ and two topological
orderings $<$ and $<'$, we have $\pre'(v)\subseteq\pre(v)$, then
$$
    v\ciph \bigl(\pre(v)\setminus\pre'(v)\bigr)\cd \pre'(v)
$$
with respect to $\omega$.
\end{prop}

\begin{proof}
Apply Corollary~\ref{cor:ordered-kernel-recovery} to the two
orderings. Since $\mathcal Q$ is order-invariant, both orderings give the same
state $\omega$. The ordering $<$ gives
$\omega_{v|\pre(v)}=Q_{v|\parents(v)}\otimes
I_{\pre(v)\setminus\parents(v)}$, whereas $<'$ gives
$\omega_{v|\pre'(v)}=Q_{v|\parents(v)}\otimes
I_{\pre'(v)\setminus\parents(v)}$. Since $<'$ is topological,
$\parents(v)\subseteq\pre'(v)$. Embedding both identities into
$\cL(\hilb_{v\cup\pre(v)})$, we obtain
$$
    \omega_{v|\pre(v)}
    =
    \omega_{v|\pre'(v)}\otimes I_{\pre(v)\setminus\pre'(v)}.
$$
By Proposition~\ref{prop:one-sided-identities-fixed-state}, this implies the
asserted quantum conditional independence.
\end{proof}

The previous proposition shows that order-invariance produces conditional
independences whenever a vertex can be moved earlier in the topological
ordering.
It yields the most restrictions when non-parent predecessors of a vertex can
be placed after that vertex in another topological ordering. The next theorem
identifies when this mechanism recovers all directed Markov restrictions.

Call a DAG $\dag$
\emph{transitive} if whenever there is a directed path $i\to\cdots\to k$, there
is also an edge $i\to k$. Equivalently, every ancestor of a vertex is one of
its parents.

\begin{thm}[Order-invariance on transitive DAGs]
\label{thm:transitive-order-invariant-implies-markov}
Let $\dag$ be transitive, and let $\mathcal Q$ be an order-invariant family of
quantum kernels with common state $\omega=\Phi_\dag(\mathcal Q)$. Then $\omega$
satisfies all three directed Markov properties {\rm(O)}, {\rm(L)} and {\rm(G)}.
\end{thm}

\begin{proof}
Fix a topological ordering $<$ and a vertex $v$. We prove the ordered Markov
restriction
$$
    v\ciph \bigl(\pre(v)\setminus\parents(v)\bigr)\cd\parents(v).
$$
Since $\dag$ is transitive, the ancestors of $v$ are precisely its parents,
and these are exactly the vertices forced to precede $v$.  Topologically order
this ancestral set, place $v$ next, and then extend to a topological ordering
$<'$ of all vertices.  Thus $\pre'(v)=\parents(v)$. Applying
Proposition~\ref{prop:order-comparison-qci} to $<$ and $<'$ gives the desired
conditional independence. Since $<$ and $v$ were arbitrary, the ordered Markov
property holds for every topological ordering. The result follows from
Theorem~\ref{thm:dagmarkeq}.
\end{proof}

The converse holds in the sense that, for non-transitive DAGs,
order-invariance of the kernel construction does not force the directed
Markov property.

\begin{prop}[The non-transitive obstruction]
\label{prop:nontransitive-obstruction}
If $\dag$ is not transitive, then there is an order-invariant family of
positive definite quantum kernels whose common state is not directed quantum
Markov with respect to $\dag$.
\end{prop}

\begin{proof}
Since $\dag$ is not transitive, there are vertices $i,j,k$ with
$i\to j\to k$ but no edge $i\to k$. Relabel the construction of
Subsection~\ref{subsec:three-qubit-example} onto this chain.  Denote its
initial state by $\sigma_{ij}$, put
$R_i=(\sigma_{ij})_i$ and
$R_{j|i}=(\sigma_{ij})_{j|i}$, and denote its final Pauli kernel by
$R_{k|j}$.  The resulting three-system state fails $k\ciph i\cd j$.

Embed these data into the parent sets of $\dag$ by taking
$Q_{i|\parents(i)}=R_i\otimes I_{\parents(i)}$,
$Q_{j|\parents(j)}
=R_{j|i}\otimes I_{\parents(j)\setminus\{i\}}$, and
$Q_{k|\parents(k)}
=R_{k|j}\otimes I_{\parents(k)\setminus\{j\}}$.
At every other vertex $v$, choose a positive definite state $\tau_v$ and use
the product kernel $\tau_v\otimes I_{\parents(v)}$.  Every topological
ordering places $i,j,k$ in that order; all remaining updates merely append
independent tensor factors.  Hence every ordering gives the same state
$\omega_{ijk}\otimes\bigotimes_{v\notin\{i,j,k\}}\tau_v$.

If this state satisfied the ordered Markov property, then at $k$ it would in
particular give $k\ciph i\cd\parents(k)$.  The additional parents are
independent tensor factors, so this condition reduces to
$k\ciph i\cd j$ for $\omega_{ijk}$, contradicting the Pauli construction.
The counterexample is robust: it works for every $0<r<1$ and
$\theta\neq0$.
\end{proof}

Thus transitivity is exactly the graph-theoretic condition under which
order-invariance of the extrinsic kernel construction forces the intrinsic
directed quantum Markov interpretation for all choices of positive definite
kernels. For non-transitive DAGs, order-invariance removes dependence on the
chosen topological ordering, but it still does not force the extrinsic kernels
to become intrinsic parent conditionals.

\section{Logarithmic characterization and reconstruction}
\label{sec:logarithmic-directed}

\subsection{Additive characterization}

The intrinsic factorization of Section~\ref{sec:dag-markov} is nonlinear
because it uses the $\star$-product.  Logarithmic conditionals provide a
complementary additive characterization.  They also lead to a canonical
operator associated with an arbitrary state, whether or not that state is
Markov.

\begin{prop}[Logarithmic Markov identity]
\label{prop:logarithmic-directed-factorization}
Let $\rho\in\mathbb S_1^+(\hilb_V)$. Then $\rho$ satisfies any one--and hence
all--of the directed Markov properties {\rm(O)}, {\rm(L)} and {\rm(G)} if and
only if
\begin{equation}\label{eq:logarithmic-markov-identity}
    \log\rho
    =
    \sum_{v\in V}
    \left(
        \log\rho_{v\cup\parents(v)}
        -
        \log\rho_{\parents(v)}
    \right)
    =
    \sum_{v\in V}h_{v|\parents(v)}(\rho),
\end{equation}
where all terms are embedded into $\cL(\hilb_V)$ and
$\log\rho_{\varnothing}=0$.
\end{prop}

\begin{proof}
Fix a topological ordering $v_1<\cdots<v_n$.  If $\rho$ is directed quantum
Markov, Proposition~\ref{prop:one-sided-identities-fixed-state} turns each
ordered Markov restriction into
$h_{v_j|V_{j-1}}(\rho)=h_{v_j|\parents(v_j)}(\rho)$.  Summing these
identities gives
\begin{equation}\label{eq:logarithmic-telescoping}
    \sum_{j=1}^n h_{v_j|V_{j-1}}(\rho)
    =
    \sum_{j=1}^n
    \bigl(\log\rho_{V_j}-\log\rho_{V_{j-1}}\bigr)
    =
    \log\rho.
\end{equation}
Equation~\eqref{eq:logarithmic-telescoping} proves
\eqref{eq:logarithmic-markov-identity}.

Conversely, take the expectation of
\eqref{eq:logarithmic-markov-identity} in $\rho$.
Equation~\eqref{eq:excess-global-information-logarithmic} gives
$gI(\dag)_\rho=0$, so Theorem~\ref{thm:entropy_relation} implies that
$\rho$ is directed quantum Markov.
\end{proof}

Thus directed quantum Markovness is linear in logarithmic coordinates: the
global Hamiltonian $-\log\rho$ is obtained by adding local directed
logarithmic conditionals. This is the noncommutative analogue of
$\log p(x_V)=\sum_{v\in V}\log p(x_v\cd x_{\parents(v)})$.

\subsection{The logarithmic candidate}
\label{subsec:logarithmic-candidate}

The right-hand side of \eqref{eq:logarithmic-markov-identity} can be evaluated
at any positive definite state, whether or not it is directed Markov.
For $\sigma\in\mathbb S_1^+(\hilb_V)$, define
\begin{equation}\label{eq:logarithmic-candidate-definitions}
    T_\dag(\sigma)
    :=
    \exp\left\{
        \sum_{v\in V}h_{v|\parents(v)}(\sigma)
    \right\}.
\end{equation}
Thus $T_\dag(\sigma)$ is the canonical operator reconstructed from the local
logarithmic conditionals of $\sigma$.  The logarithmic conditionals computed
from its own marginals need not coincide with those used in its construction.
Its relation to the input state is summarized by
\begin{equation}\label{eq:relative-entropy-log-candidate}
    gI(\dag)_\sigma
    =
    \trace\!\left(
        \sigma\bigl(\log\sigma-\log T_\dag(\sigma)\bigr)
    \right).
\end{equation}
Indeed, this follows directly from
\eqref{eq:logarithmic-candidate-definitions} and
\eqref{eq:excess-global-information-logarithmic}.
The candidate is always positive definite, but it need not have trace one.
The next result gives the precise trace bounds and identifies the case of
equality in the upper bound.

\begin{thm}
\label{thm:normalized-logarithmic-candidate}
Let $\sigma\in\mathbb S_1^+(\hilb_V)$. Then
$$
    e^{-gI(\dag)_\sigma}
    \leq
    \trace\!\left(T_\dag(\sigma)\right)
    \leq
    1.
$$
If $\trace\!\left(T_\dag(\sigma)\right)=1$, then $T_\dag(\sigma)$ satisfies
the directed Markov properties {\rm(O)}, {\rm(L)} and {\rm(G)}.
\end{thm}

\begin{proof}
Put
$c=\trace\!\left(T_\dag(\sigma)\right)$ and
$\tau=T_\dag(\sigma)/c$ so that $\log\tau=\log T_\dag(\sigma)-(\log c)I$. Klein's inequality and
\eqref{eq:relative-entropy-log-candidate} give
$$
    0
    \leq
    D(\sigma\|\tau)
    =
    gI(\dag)_\sigma+\log c,
$$
which proves the lower bound. Write $P_v=\parents(v)$. We have
$$
\log c
    =
    \trace\!\left(
        \tau\bigl(\log T_\dag(\sigma)-\log\tau\bigr)
    \right)
    =
    \sum_{v\in V}
    \trace\!\left(\tau h_{v|P_v}(\sigma)\right)
    -
    \trace\!\left(\tau\log\tau\right).
$$
Using \eqref{eq:logarithmic-candidate-definitions} and
\eqref{eq:excess-global-information-logarithmic}, we get
\begin{equation}\label{eq:normalized-candidate-gap}
\begin{aligned}
    \log c
    &=
    -gI(\dag)_\tau
    +
    \sum_{v\in V}
    \trace\!\left(
        \tau\bigl(h_{v|P_v}(\sigma)-h_{v|P_v}(\tau)\bigr)
    \right)\\
    &=
    -gI(\dag)_\tau
    -
    \sum_{v\in V}
    \left[
        D\!\left(
            \tau_{v\cup P_v}\middle\|\sigma_{v\cup P_v}
        \right)
        -
        D\!\left(
            \tau_{P_v}\middle\|\sigma_{P_v}
        \right)
    \right].
\end{aligned}
\end{equation}
Every bracket in \eqref{eq:normalized-candidate-gap} is nonnegative by
data processing under the partial trace 
\citep{lindblad1975completely}, while
$gI(\dag)_\tau\geq0$ by Theorem~\ref{thm:entropy_relation}. Therefore
\eqref{eq:normalized-candidate-gap} gives $\log c\leq0$, proving the upper
bound. If $c=1$, then $\tau=T_\dag(\sigma)$ and
\eqref{eq:normalized-candidate-gap} forces $gI(\dag)_\tau=0$.
Theorem~\ref{thm:entropy_relation} now gives the three directed Markov
properties.
\end{proof}

\begin{cor}[Fixed points of the logarithmic construction]
\label{cor:logarithmic-fixed-points}
A state $\sigma\in\mathbb S_1^+(\hilb_V)$ satisfies any one--and hence
all--of the directed Markov properties {\rm(O)}, {\rm(L)} and {\rm(G)} if and
only if $T_\dag(\sigma)=\sigma$.
\end{cor}

\begin{proof}
By \eqref{eq:logarithmic-candidate-definitions},
$T_\dag(\sigma)=\sigma$ is equivalent to
\eqref{eq:logarithmic-markov-identity}.  The conclusion follows from
Proposition~\ref{prop:logarithmic-directed-factorization}.
\end{proof}

\section{Quantum Markov equivalence}
\label{sec:markov-equivalence}
Until now all graphical statements have concerned DAGs.  We first recall the
undirected counterpart in order to compare the two settings.  For an
undirected graph $\graph$ on $V$, the notation $A\gse B\cd C$ means that every
path from $A$ to $B$ meets $C$.  A state $\rho$ is \emph{globally quantum
Markov} with respect to $\graph$ if
$$
    A\gse B\cd C
    \quad\Longrightarrow\quad
    A\ciph B\cd C.
$$
This is the undirected quantum Markov property studied in
\citet[Sections~2.4, 2.5, and~3.2]{lauritzen:zwiernik:26}.

Two directed or undirected graphs are \emph{separation equivalent} if they
encode the same separation statements. By definition, such graphs have the same
Markov states. The results below go further: for perfect orientations and for
Markov-equivalent DAGs, the corresponding entropy and logarithmic quantities
agree even away from the Markov states.

\subsection{Directed and undirected Markov properties}

Recall that the moral graph $\dag^m$ joins every pair of parents of a common
child and then removes all arrowheads.  Every separation in $\dag^m$ implies
the corresponding $d$-separation in $\dag$; hence a directed quantum Markov
state for $\dag$ is globally quantum Markov for $\dag^m$.
For a directed quantum Markov state,
Proposition~\ref{prop:logarithmic-directed-factorization} gives a parallel
operator statement: $\log\rho$ is a sum of terms supported on
$v\cup\parents(v)$ and on $\parents(v)$, both of which are complete in
$\dag^m$. Thus the logarithmic decomposition is clique-local for the moral
graph.  This observation is only one-way in general: moralization can remove
directed separation statements, and a clique-local logarithmic sum need not
imply those statements. Exact agreement arises when $\dag$ is
\emph{perfect}, meaning
that every parent set $\parents(v)$ is complete in its skeleton.

\begin{prop}\label{prop:skeleton-equivalence}
A DAG $\dag$ is separation equivalent to its skeleton $\ske(\dag)$ if and
only if $\dag$ is perfect.
\end{prop}

\begin{proof}
\citet[Proposition~3.33]{lauritzen:26}.
\end{proof}

An undirected graph is \emph{chordal} if every cycle of length at least four
has a chord.  The skeleton of a perfect DAG is chordal, and every chordal
graph $\graph$ has a perfect directed version $\dag$ with
$\graph=\ske(\dag)$.

Trees give the simplest example.  Choose a root and orient every edge away
from it.  Each non-root vertex then has exactly one parent, so the orientation
is perfect and Proposition~\ref{prop:skeleton-equivalence} shows that it is
separation equivalent to the original tree.  Consequently, a state is
directed quantum Markov for any such rooted orientation exactly when it is
globally quantum Markov for the undirected tree.  This matters because the
directed recursion, entropy criterion, and logarithmic identity can therefore
be used for the tree model for any choice of root, without changing its set
of Markov states.

For a chordal graph, its maximal cliques can be arranged in a junction tree:
the cliques containing any fixed vertex form a connected subtree.  Write
$\cliques$ for the maximal cliques, $\sep$ for the distinct intersections
labelling the junction-tree edges, and $\nu(S)$ for the number of times a
separator $S$ occurs.  This is the clique--separator notation used in
\citet[Section~3.2]{lauritzen:zwiernik:26}.  For every perfect directed
version $\dag$ of $\graph$ and every positive definite state $\rho$ we have
\begin{equation}\label{eq:chordal-directed-global-information}
\begin{aligned}
    gI(\graph)_\rho
    &:=
    \sum_{C\in\cliques}\ent(C)_\rho
    -
    \sum_{S\in\sep}\nu(S)\ent(S)_\rho
    -
    \ent(V)_\rho \\
    &=
    \sum_{v\in V}\ent(v\cd\parents(v))_\rho-\ent(V)_\rho
    =
    gI(\dag)_\rho.
\end{aligned}
\end{equation}
The middle equality is the usual clique--separator cancellation: successively
eliminating vertices in a perfect ordering cancels the nonmaximal directed
families against parent sets, leaving the maximal cliques and the junction-tree
separators.

\begin{cor}[Perfect orientations]\label{cor:perfect-orientations}
Let $\graph$ be chordal, let $\dag$ be a perfect directed version of
$\graph$, and let $\rho\in\mathbb S_1^+(\hilb_V)$.  The following statements
are equivalent:
\begin{enumerate}[\rm(i)]
\item $\rho$ is directed quantum Markov with respect to $\dag$;
\item $\rho$ is globally quantum Markov with respect to $\graph$;
\item $gI(\graph)_\rho=gI(\dag)_\rho=0$;
\item the clique--separator logarithmic identity holds:
\begin{equation}\label{eq:chordal-logarithmic-factorization}
    \log\rho
    =
    \sum_{C\in\cliques}\log\rho_C
    -
    \sum_{S\in\sep}\nu(S)\log\rho_S.
\end{equation}
\end{enumerate}
\end{cor}

\begin{proof}
Since $\dag$ is perfect, Proposition~\ref{prop:skeleton-equivalence} gives
{\rm(i)}$\Leftrightarrow${\rm(ii)}.  The equivalence with {\rm(iii)} follows
from \eqref{eq:chordal-directed-global-information} and
Theorem~\ref{thm:entropy_relation}; equivalently, it is the chordal entropy
criterion of \citet[Proposition~3.9]{lauritzen:zwiernik:26}.

The same clique--separator cancellation as in
\eqref{eq:chordal-directed-global-information}, now applied to logarithms,
gives
\begin{equation}\label{eq:chordal-directed-logarithmic-sum}
    \sum_{v\in V}h_{v|\parents(v)}(\rho)
    =
    \sum_{C\in\cliques}\log\rho_C
    -
    \sum_{S\in\sep}\nu(S)\log\rho_S.
\end{equation}
Identity~\eqref{eq:chordal-directed-logarithmic-sum} and
Proposition~\ref{prop:logarithmic-directed-factorization} give
{\rm(i)}$\Leftrightarrow${\rm(iv)}.
\end{proof}

Thus a chordal graph and any of its perfect orientations describe exactly the
same quantum Markov states.  In this case the directed entropy and logarithmic
identities reduce to their clique--separator counterparts.
Equation~\eqref{eq:chordal-logarithmic-factorization} is the appropriate
noncommutative junction-tree factorization; it becomes an ordinary product
only when the relevant operators commute.

\subsection{Equivalent DAGs}
We now compare different orientations rather than directed and undirected
graphs.  The question is when changing arrow directions leaves all
$d$-separation statements, and hence all directed quantum Markov states,
unchanged.
An \emph{unshielded collider} is an induced subgraph of the form
$\alpha\to\gamma\leftarrow\beta$.

\begin{thm}[DAG Markov equivalence]\label{thm:dagequiv}
Two DAGs $\dag_1=(V,E_1)$ and $\dag_2=(V,E_2)$ are separation equivalent if
and only if they have the same skeleton and the same unshielded colliders.
\end{thm}

\begin{proof}
This is the classical Markov-equivalence criterion
\citep{frydenberg:90,verma:pearl:90}; see also
\citet[Theorem~3.40]{lauritzen:26}.
\end{proof}

When two directed acyclic graphs $\dag_1$ and $\dag_2$ are separation
equivalent, a state $\rho$ is directed quantum Markov with respect to
$\dag_1$ if and only if it is directed quantum Markov with respect to
$\dag_2$.  We therefore also call such DAGs \emph{quantum Markov equivalent}.

The preceding conclusion concerns the zero set of $gI(\dag)_\rho$.  In fact,
both the value of this quantity and the full logarithmic candidate are
invariant throughout a Markov-equivalence class.

\begin{prop}[Equivalence-class invariance]
\label{prop:equivalence-class-invariance}
Let $\dag_1$ and $\dag_2$ be separation-equivalent DAGs on $V$.  For every
$\sigma\in\mathbb S_1^+(\hilb_V)$,
$$
    T_{\dag_1}(\sigma)=T_{\dag_2}(\sigma)
    \qquad\text{and}\qquad
    gI(\dag_1)_\sigma=gI(\dag_2)_\sigma.
$$
\end{prop}

\begin{proof}
An edge $a\to b$ is \emph{covered} if
$\parents(b)=\parents(a)\cup\{a\}$.  Any two Markov-equivalent DAGs can be
connected by a sequence of covered edge reversals
\citep[Theorem~2]{chickering:95}.  It is therefore enough to consider one
such reversal.

Such a reversal changes only the parent sets of $a$ and $b$.  Let
$P=\parents(a)$.  The two affected summands before the reversal are
$h_{a|P}(\sigma)$ and $h_{b|P\cup\{a\}}(\sigma)$; afterwards they are
$h_{b|P}(\sigma)$ and $h_{a|P\cup\{b\}}(\sigma)$.  The equality
\begin{equation}\label{eq:covered-reversal-telescoping}
    h_{a|P}(\sigma)+h_{b|P\cup\{a\}}(\sigma)
    =
    \log\sigma_{\{a,b\}\cup P}-\log\sigma_P
    =
    h_{b|P}(\sigma)+h_{a|P\cup\{b\}}(\sigma)
\end{equation}
shows that their sum is unchanged.  Hence
\eqref{eq:covered-reversal-telescoping} proves that the sum in
\eqref{eq:logarithmic-candidate-definitions} is unchanged by every covered
reversal. Therefore so is the candidate, and
\eqref{eq:relative-entropy-log-candidate} then gives the equality of the
excess global informations.
\end{proof}

For example, consider the fixed path $1-2-3$.  Its three non-collider
orientations are Markov equivalent.  For each of them,
$$
    \log T_\dag(\sigma)
    =
    \log\sigma_{\{1,2\}}
    +
    \log\sigma_{\{2,3\}}
    -
    \log\sigma_2,
    \qquad
    gI(\dag)_\sigma=I(1:3\cd2)_\sigma.
$$
The collider $1\to2\leftarrow3$ belongs to a different equivalence class and
satisfies instead
$$
    \log T_\dag(\sigma)
    =
    \log\sigma_1+\log\sigma_3
    +
    \log\sigma_{\{1,2,3\}}
    -
    \log\sigma_{\{1,3\}},
    \qquad
    gI(\dag)_\sigma=I(1:3)_\sigma.
$$
Its moral graph is complete, so every state is globally quantum Markov with
respect to the moral graph, whereas directed Markovness still requires
$1\ciph3$.  This also shows why the moral-graph implication above is generally
strict.

Proposition~\ref{prop:equivalence-class-invariance} is stronger than equality
of the corresponding sets of Markov states: it identifies the entropy defect
and the logarithmic reconstruction even for non-Markov input states. Thus
$gI(\dag)_\sigma$ and $T_\dag(\sigma)$ depend only on the
Markov-equivalence class. This conclusion is intrinsic. An
extrinsic kernel family is indexed by the parent sets of a particular DAG, and
an edge reversal changes those domains; without an additional rule
transforming the kernels, there is no corresponding invariance for the
extrinsic construction.

\section*{Acknowledgements}
\addcontentsline{toc}{section}{Acknowledgements}

PZ acknowledges the support of Ayudas Fundaci\'on BBVA a Proyectos de
Investigaci\'on Cient\'{\i}fica 2021, the Spanish Ministry of Economy and
Competitiveness grant PID2022-138268NB-I00, financed by
MCIN/AEI/10.13039/501100011033, FSE+MTM2015-67304-P, FEDER, EU.  PZ is also
affiliated with the Serra H\'unter S\`enior Program and the Barcelona School
of Economics.

\end{document}